\documentclass[conference]{IEEEtran}

\IEEEoverridecommandlockouts

\usepackage[pdftex]{graphicx}
\DeclareGraphicsExtensions{.pdf,.jpeg,.png}

\usepackage{algorithm}
\usepackage{algpseudocode}

\usepackage{etoolbox}
\usepackage{tabu}
\usepackage{setspace}
\usepackage{makecell}
\usepackage{epsfig}
\usepackage[usenames, dvipsnames]{color}
\usepackage[roman]{parnotes}

\usepackage{amsthm,amsmath,amssymb}
\usepackage{amsmath}
\usepackage{mathtools}
\usepackage{breqn}
\usepackage{amsfonts}
\usepackage{amssymb}
\usepackage{amsthm}
\usepackage{mathrsfs}
\usepackage{xparse}
\usepackage{multirow}
\usepackage{chngcntr}
\usepackage{dblfloatfix}    
\usepackage{soul} 
\usepackage{cite}
\usepackage{bm}
\usepackage{algorithmicx}
\usepackage{algpseudocode}

\algdef{SE}[DOWHILE]{Do}{doWhile}{\algorithmicdo}[1]{\algorithmicwhile\ #1}%

\makeatletter
\def\ps@IEEEtitlepagestyle{%
  \def\@oddhead{\mycopyrightnotice}%
  \def\@oddfoot{\hbox{}\@IEEEheaderstyle\leftmark\hfil}\relax
  \def\@evenhead{\@IEEEheaderstyle\thepage\hfil\leftmark\hbox{}}\relax
  \def\@evenfoot{}%
}
\def\mycopyrightnotice{%
  \begin{minipage}{\textwidth}
  \scriptsize
  Copyright~\copyright~2021 IEEE. Personal use of this material is permitted. Permission from IEEE must be obtained for all other uses, in any current or future media, including reprinting/republishing this material for advertising or promotional purposes, creating new collective works, for resale or redistribution to servers or lists, or reuse of any copyrighted component of this work in other works by sending a request to pubs-permissions@ieee.org. Accepted for publication in 2021 Military Communications Conference.
\end{minipage}
}
\makeatother

\makeatletter
\newcommand*{\algrule}[1][\algorithmicindent]{%
  \makebox[#1][l]{%
    \hspace*{.2em}
    \vrule height .75\baselineskip depth .25\baselineskip
  }
}

\newcount\ALG@printindent@tempcnta
\def\ALG@printindent{%
    \ifnum \theALG@nested>0
    \ifx\ALG@text\ALG@x@notext
    \else
    \unskip
    \ALG@printindent@tempcnta=1
    \loop
    \algrule[\csname ALG@ind@\the\ALG@printindent@tempcnta\endcsname]%
    \advance \ALG@printindent@tempcnta 1
    \ifnum \ALG@printindent@tempcnta<\numexpr\theALG@nested+1\relax
    \repeat
    \fi
    \fi
}
\patchcmd{\ALG@doentity}{\noindent\hskip\ALG@tlm}{\ALG@printindent}{}{\errmessage{failed to patch}}
\patchcmd{\ALG@doentity}{\item[]\nointerlineskip}{}{}{} 
\makeatother

\allowdisplaybreaks

\begin{document}
\title{Constrained Resource Allocation Problems in Communications: An Information-assisted Approach}  

\author{\IEEEauthorblockN{I. Zakir Ahmed  and Hamid Sadjadpour\\}
\IEEEauthorblockA{Department of Electrical and Computer Engineering\\
University of California, Santa Cruz\\
}
\and
\IEEEauthorblockN{Shahram Yousefi}
\IEEEauthorblockA{Department of Electrical and Computer Engineering\\
Queen's University, Canada\\
}\thanks{\rule[2pt]{0.97\linewidth}{0.4pt} \scriptsize \indent This work was partially sponsored by the Army Research Office and was accomplished under Grant Number W911NF-20-1-0253. The views and conclusions contained in this document are those of the authors and should not be interpreted as representing the official policies, either expressed or implied, of the Army Research Office or the U.S. Government.}}


%


\maketitle

\begin{abstract}
We consider a class of resource allocation problems given a set of unconditional constraints whose objective function satisfies Bellman's optimality principle. Such problems are ubiquitous in wireless communication, signal processing, and networking. These constrained combinatorial optimization problems are, in general, NP-Hard. This paper proposes two algorithms to solve this class of problems using a dynamic programming framework assisted by an information-theoretic measure. We demonstrate that the proposed algorithms ensure optimal solutions under carefully chosen conditions and use significantly reduced computational resources. We substantiate our claims by solving the power-constrained bit allocation problem in 5G massive Multiple-Input Multiple-Output receivers using the proposed approach.
\end{abstract}
\IEEEpeerreviewmaketitle

\newcommand{\Xmatrix}{
\begin{bmatrix}
\ddots  & 0     & 0 \\
0  & \frac{1}{{\sigma_i^2}} & 0 \\
0 & 0 & \ddots
\end{bmatrix}}
\newcommand{\Ymatrix}{
\begin{bmatrix}
\ddots  & 0  & 0 \\
0  & \frac{f(b_i)l_i}{\big(1-f(b_i)\big) \sigma_i^2} & 0 \\
0 & 0 & \ddots
\end{bmatrix}}
\newcommand{\Zmatrix}{
\begin{bmatrix}
\ddots  & 0  & 0 \\
0  & \frac{\sigma_i^2}{\sigma_n^2 + \frac{f(b_i)l_i}{\big(1-f(b_i)\big)}} & 0 \\
0 & 0 & \ddots
\end{bmatrix}}
\newcommand{\ZMmatrix}{
\begin{bmatrix}
\ddots  & 0  & 0 \\
0  & \frac{\sigma_i^2}{\sigma_n^2 + \frac{f(b_i)l_i}{\big(1-f(b_i)\big)}} + \frac{1}{p} & 0 \\
0 & 0 & \ddots
\end{bmatrix}}
\newcommand{\Fmatrix}{
\begin{bmatrix}
\ddots  & 0  & 0 \\
0  & 10 & 0 \\
0 & 0 & \ddots 
\end{bmatrix}}
\newcommand{\InvFmatrix}{
\begin{bmatrix}
\ddots  & 0 & 0 \\
0  & f(b_i)\big(1-f(b_i)\big)l_i & 0 \\
0 & 0 & \ddots
\end{bmatrix}}
\section{Introduction}
\label{Intro}
The constrained resource allocation problem involving discrete resources deals with finding an optimal solution chosen from a finite solution space. In addition, the solution needs to satisfy the constraints of the problem. These problems pose a considerable challenge to solve and, in general, are NP-Hard \cite{NPHard}. We define a class of these problems $H$ with the objective function (OF) satisfying the principle of optimality (PO) without considering the constraints \cite{Bellman}. The constraint functions are assumed to be neither convex nor linear in their decision variables. Nor are the constraints required to satisfy the linear independence constraint qualification (LICQ) \cite{Nocedal}. Examples of $H$ include bit allocation (BA) in massive Multiple-Input Multiple-Output (MaMIMO) receivers under power constraints \cite{Zakir6, Zakir7}, optimal resource selection for parameter estimation in MIMO radar \cite{MimoRadar}, managing  a large number of dynamic resources (devices and channels) to achieve changing objectives with multiple complex tradeoffs in the  Internet of Battlefield Things (IoBT) paradigm \cite{Kott}, multiple relay selection in cooperative communication \cite{Behrooz}, to name a few. 
\subsection{Previous works}
\indent The well-known methods like Dynamic Programming (DP), branch and bound (BB), integer programming (IP), and their variants are proposed to solve $H$. However, they place conditions on the constraints for optimality or near-optimality \cite{Bbound, LinShu, Zakir6, BboundVar, Fisher}. Many heuristic algorithms have been proposed to solve $H$. However, these methods extract a feasible approximate solution \cite{UpOFDMA,VLcom,Clonal}. Most of the work in the literature attempting to solve $H$ use tailored approaches specific to a given problem that guarantee either optimal or near-optimal solutions \cite{Tailored3,Tailored1,Tailored2}. More recently, machine learning (ML)-based approaches are gaining widespread popularity in solving these problems \cite{MLbased}. However, ML-based methods do not guarantee optimality. A unified framework that can solve this class of problems optimally with reduced computational resources is highly desirable.
\subsection{Our contribution}
\indent The contributions of this paper are as follows:
\begin{itemize}
\item we propose a dynamic programming framework to solve the general class of problems $H$,
\item inspired by the works of Tishby et al. we incorporate an information-theoretic measure to quantify the constraint satisfaction criteria for $H$ into the DP framework \cite{Tishby_RL}, and
\item we propose two algorithms that can solve $H$ optimally in probability that provide theoretical guarantees for near-optimality with the order of computational  complexity similar to the Viterbi Algorithm (VA).
\end{itemize}

\textit{Notations:}
The column vectors are represented as boldface small letters. The superscripts $T$ denote transpose. The term $\mathbb{R}$ indicates the  space of real numbers. We represent discrete random variable $X$ with probability mass function (PMF) $p(X)$ as $X \sim p(X)$ or simply $X$. A sequence of random variables are represented as boldfaced italics, that is $\bm{X} = \{ X_1, X_2,\cdots, X_N\}$.\\

\indent The rest of this paper is organized as follows. In Section \ref{Prob}, we describe the problem $H$. We describe the constraint satisfaction as information measure in Section \ref{Inf}. We recast the problem $H$ using the information measure in Section \ref{ModProb}. In Section \ref{OA}, we establish the theoretical guarantees for an optimal solution in probability. In Section \ref{DistEval}, we discuss the evaluation of the conditional and prior distribution that are required to compute the information measure. The bit allocation problem for MaMIMO is briefly described in Section \ref{BAprob}. The simulation results and the computational complexity analysis are discussed in Sections \ref{Sim} and \ref{CCA}, respectively, followed by the conclusion in Section \ref{Conc}.
\section{Problem setup}
\label{Prob}
We define $H$ as stated below, where $\bold{x}^*$ is the optimal solution to $\eqref{cdo_eq1}$ if it exists. 
\begin{equation}\label{cdo_eq1}
\begin{split}
&\underbrace{\text{max}}_{\substack{\bold{x}}} f(\bold{x}),\\
\text{ such that }&c_i(\bold{x}) \le \alpha_i;\text{ for }1 \le i \le Q_{I},\\
&h_j(\bold{x}) = \beta_j;\text{ for }1 \le j \le Q_E,
\end{split}
\end{equation}
where $\bold{x} = [x_1, x_2,\cdots,x_N]^T \in \mathbb{R}^N$. The OF $f(\bold{x})$ satisfies the PO without the constraint functions $c_i(\bold{x})$ and $h_j(\bold{x})$, where $\alpha_i, \beta_j \in \mathbb{R}, \forall i,j$ \cite{Bellman}. Thus, it can be written in the form $f(\bold{x}) = \sum_{i=1}^{N} b_i \phi_i(x_i)$ where $b_i \in \mathbb{R}$ are constants, $x_i \in \mathcal{X}$ can only take values from the set $\mathcal{X}$ whose cardinality is $M$, and, $\phi_i:\mathcal{X} \longrightarrow \mathcal{Y},\text{ for } 1\le i \le N$. The mapping $\phi_i$ need not be known in closed form. The constraint functions $c_i(\bold{x})$ and $h_j(\bold{x})$ are not limited to linear mappings in $x_i$, nor are they convex or need to not satisfy LICQ \cite{Nocedal}. The terms $Q_{I}$ and $Q_E$ represent the number of inequality and equality constraints, respectively.

The solution to the problem $H$ defined in \eqref{cdo_eq1} can be visualized as a finite horizon Markov decision process (MDP) which is defined using a tuple $(X,\mathcal{A},p,r,q)$, where $X$ denotes the finite set of states, $\mathcal{A}$ is the finite set of actions, $P:X \times \mathcal{A} \times X^{'} \rightarrow [0,1]$ are the state transition probabilities $p_{x,a}(x^{'})$ that a state $x^{'}$ is attained when an action $a \in \mathcal{A}$ is taken in state $x$ where $x, x^{'} \in X$. A reward $r:X \times \mathcal{A} \rightarrow \mathbb{R}$ is associated with an $a \in \mathcal{A}$ from a state $x \in X$. The prior distribution $q$ is chosen such that it is a representation of the constraints of $H$. We consider the actions $a \in \mathcal{A}$ to be deterministic given $p$ and $q$.

\newtheorem{definition}{Definition}
\begin{definition}
We define a solution (or path) $\pi = \{X_1=x_1, X_2=x_2, \cdots, X_N=x_N \}$ as a sequence of states attained as a consequence of decisions $a \in A$ taken to maximize the cumulative reward in the MDP.
\end{definition}
That is, $\pi = \{X_1=x_1, X_2=x_2, \cdots, X_N=x_N \}$, where $x_1, x_2, \cdots, x_N \in \mathcal{X}$, or simply $\pi = \{x_1, x_2, \cdots, x_N \}$. Also, we represent a path $\pi_i$ as a sequence of partially observed MDP until the stage $i$. That is $\pi_i = \{ X_1=x_1, X_2=x_2, \cdots, X_i = x_i \}$, where  $x_1, x_2  \cdots, x_i \in \mathcal{X}$, or simply $\pi_i = \{ x_1, x_2, \cdots, x_i \}$. We also write the $k^{\text{th}}$ element of the sequence $\pi$ as $\pi(k)$.
\subsection{Constraint satisfaction function}\label{csf}
We define the constraint satisfaction function (CSF) $A(\cdot)$ such that 
\begin{equation}\label{thm0_eq1}
    A(\pi)= 
\begin{cases}
    1,& \text{if } \pi \text{ satisfies all the constraints }\\ 
    &c_i(\pi), h_j(\pi) \text{ of $H$; for all } i,j,\\
    0,              & \text{otherwise.}
\end{cases}
\end{equation}
We also define $A(\cdot)$ on a partially observed sequence $\pi_k$ as $A(\pi_k) = 1$ for $1 \le k < N$ if there exists at least one forward looking subsequence $\{\pi(k+1),\pi(k+2),\cdots,\pi(N)\}$ such that $\{ \pi_k, \pi(k+1),\pi(k+2),\cdots,\pi(N) \}$ satisfies all the constraints. This is represented as
\begin{equation}\label{thm0_eq2}
    A(\pi_k)= 
\begin{cases}
    1,& \text{if there exists at least one subsequence }\\
     &\{ \pi(k+1),\pi(k+2),\cdots,\pi(N) \} \text{ defined}\\
    &\text{ above, that satisfies all the constraints }\\ 
    &c_i(\pi_k), h_j(\pi_k) \text{ for all } i,j.\\
    0,              & \text{otherwise.}
\end{cases}
\end{equation}
\section{Constraint satisfaction as an Information measure}\label{Inf}
A measure of information called Information-to-go $(\mathcal{I}_g)$ associated with a sequence that specifies cumulated information processing cost or bandwidth required to quantify the future decision and action sequence was introduced in \cite{Tishby_RL}. The measure $(\mathcal{I}_g)$ defines how many bits on average the system needs to specify the future states in an MDP (or its informational regret) with respect to the prior. This is written as
\begin{equation}\label{eq4}
\begin{split}
\mathcal{I}^{\pi_m}(X_m) &=\\
&\mathbb{E}_{p(X_{m+1}, \cdots, X_N | X_m)}\log\frac{p(X_{m+1}, \cdots, X_N | X_m)}{q(X_{m+1}, \cdots, X_N)},
\end{split}
\end{equation}
where $p(X_{m+1}, X_{m+2}, \cdots, X_N | X_m)$ is the conditional distribution of the future looking sequence given a sequence $\pi_m$, and the fixed prior $q(X_{m+1}, X_{m+2}, \cdots, X_N)$.
Inspired by \cite{Tishby_RL} we propose a modified $I_g^{\pi}$ defined in \eqref{eq5} that measures the constraint satisfaction criterion. We write
\begin{equation}\label{eq5}
\begin{split}
I_g^{\pi_m}(X_m) &\triangleq \\
&\mathbb{E}_{p(X_{m+1}, \cdots, X_N | X_m)}\log\frac{p(X_{m+1}, \cdots, X_N | X_m)}{q(X_{m+1}, \cdots, X_N| X_m)}.
\end{split}
\end{equation}
Effectively, the term $I_g^{\pi_m}(X_m)$ denotes the KL divergence between the distribution of future looking sequence $\{ X_{m+1}, \cdots, X_N \}$ given $X_m$ with respect to the known prior conditional distribution of the successive future states $q(X_{m+1}, X_{m+2}, \cdots, X_N | X_m)$. The $I_g^{\pi}(X_m)$ can be thought of as the information processing cost in bits to ensure constraint satisfaction in pursuing a partially observed path $\pi_m$ going into the indefinite future with respect to the known conditional prior $q(X_{m+1}, X_{m+2}, \cdots, X_N | X_m)$.\\
\indent Intuitively, when $I_g^{\pi_m}(X_m) \approx 0$ implies least information is required to pursue the path $\pi_m$ to satisfy the CSF $A(\pi_m)$. On the other hand, a large value of $I_g^{\pi_m}(X_m)$ implies maximum information required to make the decision (or inability to make a decision) to see if the CSF $A(\pi_m)$ is satisfied when pursuing the path $\pi_m$.\\ 
\indent We write the conditionals $p(X_{m+1}, X_{m+2}, \cdots, X_N | X_m)$ as simply $p$ and the conditional priors $q(X_{m+1}, X_{m+2}, \cdots, X_N | X_m)$ as  $q$ for compact representation. The details pertaining to the evaluation of the conditionals $p$ and the prior $q$ are discussed in Section \ref{DistEval}.
\section{Problem setup with information measure}\label{ModProb}
Using the definitions and notations defined in Section \ref{Prob} and \ref{Inf} we rewrite the problem $H$ as
\begin{equation}\label{thm0_eq0}
\begin{split}
&\underbrace{\text{max}}_{\substack{\pi; A(\pi) > 0}} f^{\pi}(\bm{X}),
\end{split}
\end{equation} 
where $f^{\pi}(\bm{X}) = \sum_{i=1}^N b_i\phi_i(X_i)$ for path $\pi$.\\
\indent Decoupling the constraints from \eqref{thm0_eq0} and absorbing the same into \eqref{eq5}, the class of problems $H$ defined using \eqref{cdo_eq1} can be recast to minimize $I_g^{\pi}(\bm{X})$ that ensures the constraint satisfaction criterion and at the same time maximize the reward $f^{\pi}(\bm{X})$. We define a functional $G^{\pi}$ as \cite{Tishby_RL}
\begin{equation}\label{eq6}
\begin{split}
G^{\pi}(\bm{X},\beta) \triangleq  I_g^{\pi}(\bm{X}) - \beta f^{\pi}(\bm{X}),
\end{split}
\end{equation}
where $\beta$ is the Lagrange multiplier. The modified problem is written as
\begin{equation}\label{eq7}
\begin{split}
G^{\pi^*}(\bm{X},\beta) &\triangleq \underbrace{\text{min}}_{\substack{\pi}}\Big\{ I_g^{\pi}(\bm{X}) - \beta f^{\pi}(\bm{X}) \Big\},\text{ or}\\
\pi^* &= \underbrace{\text{argmin}}_{\substack{\pi}}\Big\{ I_g^{\pi}(\bm{X}) - \beta f^{\pi}(\bm{X}) \Big\}.
\end{split}
\end{equation}
Here $\pi^*$ is the optimal solution to \eqref{cdo_eq1} in probability.

\begin{definition}\label{def2}
  We say that the solution $\pi^{\beta_o}$ is close to $\pi^*$ in probability when $Pr\Big\{ \left\Vert\ \pi^{\beta_o} - \pi^* \right\Vert_2 \le \epsilon \Big\} \ge 1-\delta,\text{ for }\beta_o \in (\beta_L,\beta_U)$, where $\epsilon, \delta$ are small numbers close to zero.
\end{definition}
\vspace{-0.2in}
\section{Optimality Analysis}\label{OA}
If \eqref{eq7} satisfies the Bellman's optimality criterion, we can use the dynamic programming framework to solve \eqref{eq7}. We say that the value function $G^{\pi}(\bm{X},\beta)$ is said to satisfy Bellman's principle of optimality (PO) under the following conditions \cite{Bellman,Takashi,sdp,BellFail}.\\
\indent \textit{(1)} The value function $G^{\pi}(\bm{X},\beta)$ can be broken down into two parts consisting of an immediate reward component (subproblem) and a scaled (discounted) future value function for a given $\beta$.\\
\indent \textit{(2)} The subsolution $\pi^*_k$ of the optimal solution $\pi^*$ obtained by solving an incompletely observed MDP are themselves optimal solutions for their subproblems. This is illustrated below.\\

If $\pi^* = \underbrace{\text{argmin}}_{\substack{\{\pi(i)\}_{i=1}^{N}}} G^{\pi}(\bm{X},\beta)$ for some $\beta = \beta_o$, and if we can express $G^{\pi_k}(X_k,\beta_o) = H^{\pi_k}(X_k,\beta_o) + G^{\pi_{k+1}}(X_{k+1},\beta_o)$, where $H^{\pi_k}(X_k,\beta_o)$ is the subproblem defined based on the partial observation of the MDP until stage $k$, and $G^{\pi_{k+1}}(X_{k+1},\beta_o)$ is the future value function then we have
\begin{equation}\label{eq7a}
\begin{split}
&\pi^* = \underbrace{\text{argmin}}_{\substack{\{\pi(i)\}_{i=1}^{N}}} \Big\{H^{\pi_k}(X_k,\beta_o) + G^{\pi_{k+1}}(X_{k+1},\beta_o) \Big\}\\
&= \underbrace{\text{argmin}}_{\substack{\{\pi(i)\}_{i=1}^{k}}} \Big\{ H^{\pi_k}(X_k,\beta_o) + \underbrace{\text{argmin}}_{\substack{\{\pi(i)\}_{i=k+1}^{N}}} G^{\pi_{k+1}}(X_{k+1},\beta_o) \Big\}.
\end{split}
\end{equation}
Observing \eqref{eq7a}, we say that $\underbrace{\text{argmin}}_{\substack{\pi}} G^{\pi}(\bm{X},\beta)$ satisfies PO if the solution to the subproblem can be written as $\pi^*_k = \underbrace{\text{argmin}}_{\substack{\{\pi(i)\}_{i=1}^{k}}} \Big\{ H^{\pi_k}(X_k,\beta_o) \Big\}$. The subsolution $\pi^*_k$ is part of the optimal solution $\pi^*$.
 
We first show that the problem $H$ in \eqref{cdo_eq1} does not satisfy the principle of optimality (PO) using theorem \ref{thm0}, later we  argue that the modified problem \eqref{eq7} indeed satisfies PO for some $\beta_o \in (\beta_L,\beta_U)$.

\newtheorem{theorem}{Theorem}
\begin{theorem}\label{thm0}
The problem $H$ described using \eqref{cdo_eq1} does not satisfy the PO. 
\end{theorem}
\begin{proof}
From \eqref{cdo_eq1}, it is easy to see that $f^{\pi_m}(X_m) = b_m \phi_m (X_m) + f^{\pi_{m+1}}(X_{m+1})$. Using this recursion, we can write the value function in \eqref{thm0_eq0} as
\begin{equation}\label{thm0_eq3}
\begin{split}
f^{\pi}(\bm{X}) &= f^{\pi_1}(X_1),\nonumber
\end{split}
\end{equation}
\begin{equation}\label{thm0_eq3}
\begin{split}
&= \psi^{\pi_k} + f^{\pi_{k+1}}(X_{k+1}),
\end{split}
\end{equation}
where $\psi^{\pi_k} = \sum_{i=1}^k b_i \phi_i (X_i = \pi(i))$. Given that $\pi^* = \underbrace{\text{argmax}}_{\substack{\pi; A(\pi) > 0}} f^{\pi}(\bm{X})$ we say that $f^{\pi}(\bm{X})$ satisfies PO if the sequence of subsolutions $\pi^*_k$ to the subproblems $\underbrace{\text{argmax}}_{\substack{\pi_k; A(\pi_k) > 0}} \psi^{\pi_k}$ is part of the optimal solution $\pi^*$ for all $1 \le k \le N$.\\
\indent However if we have an infeasible solution $\hat{\pi}$ such that $f^{\hat{\pi}}(\bm{X}) > f^{\pi^*}(\bm{X})$, and $A(\hat{\pi}) = 0$ but the solution $\hat{\pi}$ satisfies the CSF $A(\hat{\pi}_k) > 0$ at some intermediate stage $k$, then the subproblem $\underbrace{\text{argmax}}_{\substack{\pi_k; A(\pi_k) > 0}} \psi^{\pi_k}$ will not pick the optimal sequence $\pi^*$ going forward into the future stages beyond $k$. This scenario is a consequence of placing no conditions on the objective and the constraint functions of $H$. Thus the solution obtained by solving a sequence of subproblems $\underbrace{\text{argmax}}_{\substack{\pi_k; A(\pi_k) > 0}} \psi^{\pi_k}$ will be different from $\pi^*$.
\end{proof}

It can be shown \cite{Tishby_RL} that \eqref{eq7} can be expressed recursively as 
\begin{equation}\label{thm1_eq3}
\begin{split}
G^{\pi_k}(X_k,\beta) &\triangleq  H^{\pi_k}(X_k,\beta) + G^{\pi_{k+1}}(X_{k+1},\beta),
\end{split}
\end{equation}
where 
\begin{equation}\label{thm1_eq5}
\begin{split}
H^{\pi_k}(X_k,\beta) &= D^{\pi_k} - \beta \psi^{\pi_k}(X_k),\\
D^{\pi_k} &\triangleq D_{KL}^{\pi_k}(p(X_1,\cdots,X_{k})||q(X_1,\cdots,X_{k})),\\
&\triangleq \sum_{i=1}^{k} \mathbb{E}_{p(X_{i+1}|X_i=\pi(i))} \log \frac{p(X_{i+1}|X_i=\pi(i))}{q(X_{i+1}|X_i=\pi(i))},\\
\psi^{\pi_k}(X_k) &= \sum_{i=1}^{k} b_i \phi_i(X_i=\pi(i)).
\end{split}
\end{equation}

In \eqref{eq7}, it can be seen that a larger value of $\beta$ emphasizes the reward against the constraint satisfaction criteria $I_g^{\pi}(\bm{X})$,  whereas a lower value of $\beta$ ensures the constraint satisfaction criteria has prominence over the reward. It can also be shown that the reward $f^{\pi}(\bm{X})$ vs. the information-to-go $I_g^{\pi}(\bm{X})$ for different values of $\beta$ has a monotonic non-decreasing relationship \cite{Thomas,Tishby_IB,Tishby_RL}. Keeping in mind the above points, it is easy to show that if $\pi^* = \underbrace{\text{argmin}}_{\substack{\{\pi(i)\}_{i=1}^{N}}} G^{\pi}(\bm{X},\beta_o)$ and $\pi^{\beta_o} = \underbrace{\text{argmin}}_{\substack{\{\pi(i)\}_{i=1}^{N}}} \Big\{ H^{\pi_N}(X_N,\beta_o) \Big\}$, then there always exists a range of $\beta$'s such that $\beta_o \in (\beta_L,\beta_U)$, where $Pr\Big\{ \left\Vert\ \pi^{\beta_o} - \pi^* \right\Vert_2 \le \epsilon \Big\} \ge 1-\delta$. Here $\epsilon$ and $\delta$ are small numbers close to zero. Hence implying that \eqref{eq7} indeed satisfies Bellman's PO for this range of $\beta$.

A trellis based VA can be used to find the optimal solution $\pi^*$ (in probability) to \eqref{eq7} \cite{Viterbi}. This would necessitate the computation of the path metric $PM_{\pi_{m+1}}$ at stage $m+1$ as
\begin{equation}\label{eq10a}
\begin{split}
PM_{\pi_{m+1}} &= \mathbb{E}_{p(X_{m+1}|X_{m}=\pi(m))} \log \Bigg[ \frac{p(X_{m+1}|X_{m}=\pi(m))}{q(X_{m+1}|X_{m}=\pi(m))}\Bigg]\\
&- \beta b_i \phi_{m}(X_{m+1}),
\end{split}
\end{equation}
for a path $\pi_m$ that is incident on the node $x_j \in  \mathcal{X}$ at stage $m+1$ of the trellis structure of the VA; and then select the path that has a minimum value among them. This is the well known Add-Compare-Select (ACS) operation in the VA.

A description of $p(X_{m+1}|X_{m})$ and $q(X_{m+1}|X_{m})$ at every stage of the trellis will suffice to compute the path metric in \eqref{eq10a}.

\section{Evaluating the distributions $p$ and $q$}\label{DistEval}
In this section we will discuss the methods to evaluate the distributions $p(X_{m+1}|X_{m})$ and $q(X_{m+1}|X_{m}), \forall m \in (1,N-1)$.
 
\subsection{Evaluation of the priors $q$}\label{qeval}
To evaluate the conditional priors $q(X_{t+1}=x_i | X_{t}=x_j) \forall t \in (1,N); x_i,x_j \in \mathcal{X}$, we sample a set of $K$ solutions from the exhaustive search space $B_{set}$ of problem $H$ such that they satisfy $A(\pi) > 0$. We then identify $N_1$ solutions $\{ \pi^i \}_{i=1}^{N_1}$ out of the $K$ selected solutions that have maximum reward, that is $f(\pi^1) \ge f(\pi^2), \ge \cdots \ge f(\pi^{N_1}) > f(\pi^{N_1+1}) \ge \cdots \ge f(\pi^K)$. Using these $N_1$ subset of solutions we evaluate
\begin{equation}\label{eq13a}
\begin{split}
q(X_{t+1} &= x_i | X_{t}=x_j) = \frac{F(\{ X_{t+1}=x_i | X_{t}=x_j \})}{N_1}\\
&\forall t \in (1,N) ; x_i \in \mathcal{X},
\end{split}
\end{equation}
where $F(\{ X_{t+1}=x_i | X_{t}=x_j \})$ returns the number of times the event $\{ X_{t+1}=x_i | X_{t}=x_j \}$ occur among the $N_1$ solutions. It follows that when $K \rightarrow |B_{set}|$, and for a small $N_1$ we have $q(\pi^*) \rightarrow 1$.

\subsection{Evaluation of the conditional $p$}\label{peval}
We describe two methods to evaluate the conditional distribution $p(X_{t+1}|X_t)$. In the first approach the conditionals $p(X_{t+1}|X_t)$ are derived at stage $t$ of the trellis traversal using the constraints $c_i(\cdot)$, $h_j(\cdot)$, the starting distribution of states $q(X_1)$, and the path metrics $PM_{\pi_t}$. The evaluation of the conditionals $p$ for the BA problem in MaMIMO is illustrated in  Section \ref{BAprob}.\\
\indent In the second approach, we make use of the well known iterative Blahut-Arimoto algorithm to obtain $p(X_{t+1}|X_t)$ at stage $t$ of the trellis traversal \cite{Thomas,Tishby_IB}. It can be shown that by taking derivative of $G^{\pi}(\bm{X},\beta)$ with respect to $\pi$ and then setting the gradient of $G^{\pi}$ to $0$, the equation \eqref{eq7} satisfies the equations shown below \cite{Tishby_RL,Tishby_IB,BArd, BAcap}.
\begin{equation}\label{eq14}
\begin{split}
&p^{(k)}(X_{t} = x_i) =\\
&\sum_{x_j \in \mathcal{X}}p(X_{t-1}=x_j)p^{(k-1)}(X_{t}=x_i|X_{t-1}=x_j),\\
&p^{(k)}(X_{t}=x_i|X_{t-1}=x_j) = \\
&\frac{p^{(k)}(X_t=x_i)\exp (-\beta G^{\pi_{t-1}}(X_t,\beta))}{\sum_{x_l \in \mathcal{X}}p^{(k)}(X_t=x_l)\exp (-\beta G^{\pi_{t-1}}(X_t,\beta))},
\end{split}
\end{equation}
where $k$ is the iteration number. It is also worth noting that the problem \eqref{eq7} has an analogy to the variant of the rate-distortion problem in information theory. That is, the \eqref{eq7} can be visualized as
\begin{equation}\label{RDprob}
\begin{split}
\underbrace{\text{min}}_{\substack{p(X_{t+1}|X_t)}}\Big\{ I_g^{\pi}(X_t) \Big\}\text{ such that }f^{\pi}(X_{t}) \ge D,
\end{split}
\end{equation}
where $D$ is some minimum reward that needs to be guaranteed. The solution to \eqref{RDprob} is the same set of self-consistent equations described in \eqref{eq14}.\\
\indent We propose two algorithms based on the way the conditional $p$ is constructed above. In the first variant, the conditional $p$ is evaluated using the constraints of the problem $H$. We call this Information-assisted DP (IADP-specific). In the second variant, the conditional $p$ is derived using the well known Blahut-Arimoto algorithm (BAA) \cite{Tishby_IB}. We call this algorithm IADP-BAA. Both these algorithms are traditional VA frameworks that evaluate the path metrics as defined in \eqref{eq10a} \cite{Viterbi}. However, an optimal solution is guaranteed for some $\beta_o \in (\beta_L,\beta_U)$. The proposed algorithms are run for different values of $\beta$ chosen using a binary search (BS) method \cite{Ndale}.
\section{ADC Bit Allocation for MaMIMO}\label{BAprob}
In this section, we very briefly describe the problem of BA for MaMIMO. The ADC BA problem is to assign the number of bits to be used by Variable-Resolution ADCs on different Radio Frequency (RF) paths of the MaMIMO receivers. An optimal BA ensures that the performance of the receiver is maximized under a non-linear power constraint. It is to be noted that the OF is non-linear. In \cite{Zakir7}, the authors reduce this to a problem in $H$, which is described as
\begin{equation}\label{cdo_sim_ba1}
\underbrace{\text{max}}_{\substack{\{x_i\}_{i=1}^N;A(\bold{x}) > 0}}\Big\{\sum_{i=1}^{N} \frac{a_i^2}{b_i^2+d_i2^{x_i}}\Big\},
\end{equation}
where $a_i$, $b_i$, and $d_i$ are constants $\in \mathbb{R}$ that represent channel singular value, noise power, and coefficient of quantization noise due to bit allocation $x_i$ on the $i^{th}$ RF path, respectively. Here $N$ is the number of RF paths in the receiver. The bits $x_i$ can take values from the set $\mathcal{X} = \{1,2,3,4 \}$. The CSF $A(\bold{x}) > 0$ iff the power constraint $\sum_{i=1}^{N} 2^{x_i} \le P_b$, and bit-ordering  constraint $x_1 \ge x_2  \ge \cdots  \ge x_N$ are satisfied. The total ADC power budget is $P_b$. Hence we have
\begin{equation}\label{cdo_sim3}
    A(\bold{x}) = 
\begin{cases}
    1,\text{ if } &\sum_{i=1}^{N} 2^{x_i} \le P_b,\\
       &x_1 \ge x_2  \ge \cdots  \ge x_N.\\
    0, &\text{ Otherwise}.\\
\end{cases}
\end{equation}
\subsection{Problem specific evaluation of the conditional $p$}\label{pspec}
We define $p(X_{t+1}|X_t)$ between the stages $t$ and $t+1$ for a given path $\pi_t$ based on the two constraints in \eqref{cdo_sim3}. We know the elements of the path $\pi_t$ for stages $1, \cdots t$. Thus we write
\begin{equation}\label{eq12}
\begin{split}
&p(X_{t+1} = x_i |X_t) =\\
&\frac{  S\bigg(P_b - \Big( \sum_{k=1}^t 2^{x_k} + 2^{x_i} \Big) \bigg) + n_i}{\sum_{x_j \in \mathcal{X}} \Bigg[ S\bigg( P_b - \Big( \sum_{k=1}^t 2^{x_k} + 2^{x_j} \Big) \bigg) + n_j \Bigg] } \forall x_i \in \mathcal{X},
\end{split}
\end{equation}
where $S(x) = \frac{1}{1+e^{-x}}$ is a sigmoid function that bounds domain of $S$ in $[0,1]$ for $x \in (-\infty,\infty)$. It is easy to see that the term $P_b - \big( \sum_{k=1}^t 2^{x_k} + 2^{x_i} \big)$ represents the residual power available for the path $\pi_t$ to ensure the power constraint is satisfied. The larger the term $P_b - \big( \sum_{k=1}^t 2^{x_k} + 2^{x_i} \big)$, the greater the chance of satisfying the power constraint. The condition $P_b - \big( \sum_{k=1}^t 2^{x_k} + 2^{x_i} \big) \le 0$ indicates that the power budget is exhausted for the path $\pi_t$. The normalization term in the denominator of \eqref{eq12} ensures that $\sum_{x_i \in \mathcal{X}} p(X_{t+1} = x_i |X_t) = 1$. In addition, we add noise $n_i \sim \mathcal{N}(0,\,\sigma^{2})$ with a very small variance $\sigma^{2}$ to ensure randomness in the distribution. The probabilities $p(X_{t+1} = x_i |X_t)$ can be efficiently computed on the fly for the path $\pi_t$ at stage $t$ during the trellis traversal in VA. The constraint $x_1 \ge x_2  \ge \cdots  \ge x_N$ is taken care when
\begin{equation}\label{eq12a}
\begin{split}
p(X_{t+1} = x_i |X_t = x_j) = 0 \text{ when } x_i < x_j ; \forall x_i,x_j \in \mathcal{X}.
\end{split}
\end{equation} 

\section{Simulations} 
\label{Sim}
We use the proposed algorithms to analyze the BA problem in massive MIMO described in Section \ref{BAprob}. We consider the number of RF paths $N=8$ \cite{Zakir6}. We set the power budget $P_b = 32$, which is the normalized power spent on having 2-bit ADCs on all the RF paths. We sweep the value of $\beta \in [0,10]$ in steps of $0.01$ for analysis for each of the proposed methods. The solution $\pi^{\beta}$ obtained for each $\beta$ with IADP-specific and IADP-BAA is shown in the Table \ref{tab1}. A plot of the trade-off curve between the reward $f^{\pi}(\bm{X})$ and the CSF criterion (Information-to-go) $I_g^{\pi}(\bm{X})$ for various values of $\beta$ are shown in the Fig \ref{fig3}. It can be seen that both the algorithm achieve optimal solution when $\beta \in [0.08,4.3]$ for IADP-specific, and when $\beta \in [0.02,4.31]$ for IADP-BAA. This observation corroborates our theoretical analysis as discussed in Section \ref{OA}. We also use a nonlinear BB (NLBB) algorithm with branching and pruning based on dominance and constraint satisfaction to solve the BA problem \cite{Bbound}. The solutions obtained with the proposed methods, NLBB, and exhaustive search (ES) are shown in Table \ref{tab2}. 

\begin{table}[t]
\caption{\scriptsize The solution $\pi^{\beta}$, reward and the power for various values of $\beta$ using the proposed Algorithms.}
\label{tab1}
\begin{center}
\begin{sc}
\resizebox{\columnwidth}{!}{%
\begin{tabular}{ |c|c|c|c|}
\hline
\tiny $N=8$ & \tiny Solution $\pi^{\beta}$ (IADP-Specific) & \tiny Reward & \tiny Power (normalized) \\ 
\hline
\tiny $\beta = [0,0.07]$  & \tiny $\{4,1,1,1,1,1,1,1\}$ & \tiny $17.543$ & \tiny $30$ \\
\hline
\tiny $\beta = [0.08,4.3]$  & \tiny $\{4,2,1,1,1,1,1,1\}$ & \tiny $18.0081$ & \tiny $32$ \\
\hline
\tiny $\beta = [4.31,10.0]$  & \tiny $\{4,4,4,4,4,4,1,1\}$ & \tiny $25.6008$ & \tiny $100$ \\
\hline
\end{tabular}%
}
\end{sc}
\end{center}
\begin{center}
\begin{sc}
\resizebox{\columnwidth}{!}{%
\begin{tabular}{ |c|c|c|c|}
\hline
\tiny $N=8$ & \tiny Solution $\pi^{\beta}$ (IADP-BAA) & \tiny Reward & \tiny Power (normalized) \\ 
\hline
\tiny $\beta = [0,0.01]$  & \tiny $\{4,1,1,1,1,1,1,1\}$ & \tiny $17.543$ & \tiny $30$ \\
\hline
\tiny $\beta = [0.02,4.31]$  & \tiny $\{4,2,1,1,1,1,1,1\}$ & \tiny $18.0081$ & \tiny $32$ \\
\hline
\tiny $\beta = [4.32,10.0]$  & \tiny $\{4,4,4,4,4,4,1,1\}$ & \tiny $25.6008$ & \tiny $100$ \\
\hline
\end{tabular}%
}
\end{sc}
\vskip 0.1in
\end{center}
\end{table}
\begin{table}[t]
\caption{\scriptsize The solution $\pi$ obtained using IADP-Specific, IADP-BAA, NLBB, and ES Algorithms.}
\label{tab2}
\begin{center}
\begin{sc}
\resizebox{\columnwidth}{!}{%
\begin{tabular}{ |c|c|c|c|}
\hline
\tiny Algorithm & \tiny Solution $\pi$  & \tiny Reward & \tiny Power (normalized) \\ 
\hline
\tiny IADP-Specific  & \tiny $\{4,2,1,1,1,1,1,1\}$ & \tiny $18.0081$ & \tiny $32$ \\
\hline
\tiny IADP-BAA  & \tiny $\{4,2,1,1,1,1,1,1\}$ & \tiny $18.0081$ & \tiny $32$ \\
\hline
\tiny NLBB  & \tiny $\{4,1,1,1,1,1,1,1\}$ & \tiny $17.543$ & \tiny $30$ \\
\hline
\tiny Exhaustive search  & \tiny $\{4,2,1,1,1,1,1,1\}$ & \tiny $18.0081$ & \tiny $32$ \\
\hline
\end{tabular}%
}
\end{sc}
\end{center}
\end{table}

\begin{figure}[t!]
\centering
\includegraphics[scale=0.4]{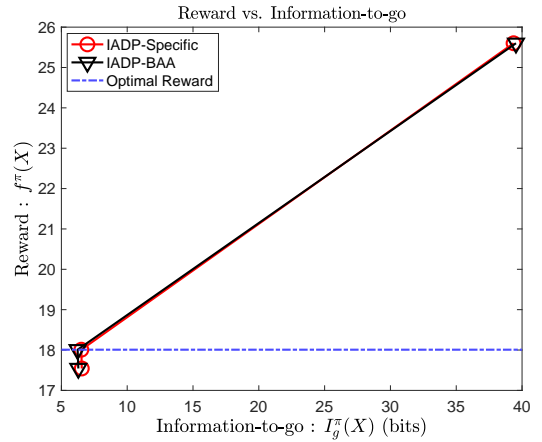}
\caption{\small Reward vs. Information-to-go for BA MaMIMO problem for $N=8$.}
\label{fig3}
\end{figure}

\section{Computational Complexity Analysis}
\label{CCA}
This section compares the computational complexity (CC) of the proposed IADP-Specific and IADP-BAA methods, NLBB, and the ES algorithms.\\
\indent \textit{(i) Exhaustive search (ES): } The total number of solutions in the $B_{\text{set}} = M^N$, and hence it has a CC of $O(M^N)$.\\ 
\indent \textit{(ii) NLBB: } Obtaining an exact solution using NLBB has a worst-case computational complexity similar to that of the ES, which is $O(M^N)$.\\ 
\indent \textit{(iii) IADP-Specific: } The total number of ACS evaluations for an $M$ state trellis with horizon length of $N$ is $NM^2$, and a total of $(N-1)M^2$ evaluations are required for $p(X_{m+1}|X_m), \forall m \in (1,N-1)$. For the evaluation of the priors $q$ as discussed in Section \ref{qeval} we need $K$ solutions to be sampled, and hence the complexity is $K$. It can be shown that for a BS algorithm with an exit range threshold $T_{\text{Range}} = \beta_U - \beta_L$, and for a maximum value of beta $\beta_{\text{Max}}$, the average number of searches required is $\log_2(\frac{\beta_{\text{Max}}}{T_{\text{Range}}})$. Hence the overall computations required for IADP-specific is
\begin{equation}\label{eq_cc1}
\begin{split}
T_{\text{IADP-specific}} = (NM^2 + (N-1)M^2) \log_2\Big(\frac{\beta_{\text{Max}}}{T_{\text{Range}}}\Big) + K.
\end{split}
\end{equation}
\indent \textit{(iv) IADP-BAA: } For the IADP-BAA, the only difference compared to IADP-Specific is the computation of the conditionals $p$. A total of $N_{\text{iter}}(N-1)M^2$ computation is required for $p$, where $N_{\text{iter}}$ is the average number of iterations required for the BAA to achieve the required convergence. Thus we have
%
%
\begin{equation}\label{eq_cc2}
\begin{split}
T_{\text{IADP-BAA}} = (NM^2 + N_{\text{iter}}(N-1)M^2) \log_2\Big(\frac{\beta_{\text{Max}}}{T_{\text{Range}}}\Big) + K.
\end{split}
\end{equation}
It can be observed that both the proposed IADP-Specific and IADP-BAA algorithms have overall complexity of $O(NM^2)$.\\
\indent Although the complexity in terms of the oracle notation is the same for both IADP-Specific and IADP-BAA, the total number of arithmetic operations required for IADP-BAA is greater than IADP-Specific because of the iterative nature of BAA. However, both methods have the same order of complexity as that of the Viterbi algorithm.

\section{Conclusion} 
\label{Conc}
We turn the class of constrained discrete resource allocation problems ubiquitous in wireless communication and signal processing to an unconstrained problem using an information-theoretic measure.  A dynamic programming framework assisted by this information measure is proposed to solve these problems. The proposed methodologies provide a common framework to address this class of problems. We provide theoretical analysis to establish near-optimality guarantees and show that the computational complexity order of the proposed algorithms is as good as the Viterbi algorithm.
 
\bibliographystyle{IEEEtran}
\bibliography{MilCom21_BibTexFile}
%
\end{document}